%% file: ms.tex
\documentclass[letterpaper, 10 pt, journal, twoside]{IEEEtran}  

\IEEEoverridecommandlockouts                              
\usepackage{graphicx} 
\graphicspath{{figures/}}
\usepackage{amsmath} 
\usepackage{amssymb}  
\usepackage{color}
\usepackage{gensymb}
\usepackage{pgfplots}
\usepackage[absolute,overlay]{textpos}
\usepackage{mathtools}
\usepackage{algorithm}
\usepackage{algpseudocode}
\usepackage{bbold}
\usepackage{verbatim}
\usepackage{multirow}
\usepackage{float}
\usepackage{tikz}

\usepackage{amsthm}
\usepackage{dsfont}
\usepackage{cite}
\mathtoolsset{showonlyrefs} 
\let\labelindent\relax
\usepackage{enumitem}
\usepackage{subcaption}

\newtheorem{definition}{Definition}
\newtheorem{lemma}{Lemma}

\newtheorem{remark}{Remark}

\newtheorem{assumption}{Assumption}

\newtheoremstyle{break}
  {}
  {}
  {\itshape}
  {}
  {\bfseries}
  {.}
  {\newline}
  {}

\theoremstyle{break}
\newtheorem{theorem}{Theorem}

\newenvironment{myproof1}{\noindent {\bf Proof of Theorem \ref{thm:lin}}}{\hfill$\blacksquare$}
\newenvironment{myproof2}{\noindent {\bf Proof of Theorem \ref{thm:polypoa}}}{\hfill$\blacksquare$}
\newenvironment{myproof3}{\noindent {\bf Proof of Theorem \ref{thm:routing}}}{\hfill$\blacksquare$}

\definecolor{Granata}{rgb}{0.64,0,0} 

\newcommand{\dar}{\textcolor{black}}

\newcommand{\fra}{\textcolor{black}}

\newcommand\mc[1]{\mathcal{#1}}
\newcommand\mb[1]{\mathbb{#1}}
\newcommand\R{\mathbb{R}}

\renewcommand\i{^i}
\renewcommand\j{^j}

\newcommand{\N}{M}
\newcommand\NE{_\textup{N}}
\newcommand\WE{_\textup{W}}
\newcommand\SO{_\textup{S}}
\newcommand\s{\Sigma}
\newcommand\poa{\textup{PoA}}

\title{\LARGE \bf On the Efficiency of Nash Equilibria in \dar{Aggregative} Charging Games
}

\author{Dario Paccagnan, Francesca Parise and John Lygeros
\thanks{This work was supported by the SCCER FEEB\&D, by the European Commission project DYMASOS (FP7-ICT 611281), and by the SNSF grant numbers P1EZP2 172122 and P2EZP2 168812. D. Paccagnan and J. Lygeros are with the Automatic Control Laboratory, ETH Z\"{u}rich, Switzerland. Email:
        {\{\tt\footnotesize dariop,lygeros\}@control.ee.ethz.ch.}
F. Parise is with the Laboratory for Information and Decision Systems, MIT, Cambridge, MA, USA. Email: {\tt parisef@mit.edu}.
        }}

\begin{document}

\maketitle
\thispagestyle{empty}
\pagestyle{empty}

\begin{abstract}
Several works have recently suggested to model the problem of coordinating the charging needs of a fleet of electric vehicles as a game, and have  proposed distributed algorithms to coordinate the vehicles towards a Nash equilibrium of such game. However, Nash equilibria have been shown to posses desirable system-level properties only in  simplified cases. 
In this work, we use the concept of price of anarchy to analyze the inefficiency of Nash equilibria when compared to the social optimum solution. More precisely, we show that i)  for linear price functions depending on all the charging instants, the price of anarchy converges to one as the population of vehicles grows; ii) 
for price functions that depend only on the instantaneous demand, the price of anarchy  converges to one if the price function takes the form of a positive pure monomial; \fra{iii) for general classes of price functions, the asymptotic price of anarchy can be bounded.}
For finite populations, we additionaly provide a bound on the price of anarchy as a function of the number vehicles in the system. We support \mbox{the theoretical findings
by means of numerical simulations.}
\end{abstract}
\vspace*{-5mm}
\section{Introduction}
In the last decade we have witnessed a profound change in the way  energy systems are operated.  
A new paradigm called demand response is emerging, according to which the energy requirements of a population of users are tuned, by means of incentives, to account for the operational needs of the power grid \cite{albadi2007demand}. Previous works \cite{ma2013decentralized,gan2013optimal} have suggested to model these demand response methods as a game. Therein each player represents a user that needs to optimize his energy consumption over a given period of time, with the objective of minimizing his electricity bill. What couples the users, and thus makes the charging problem a game, is the assumption that the energy price depends at every instant of time on the sum of the energy demand of the whole population. 
The seminal paper \cite{ma2013decentralized} shows that the (unique) Nash equilibrium of such game has desirable properties from the standpoint of the grid operator, in the case of large and homogeneous populations. Under these assumptions, \cite{ma2013decentralized} shows that the equilibrium is socially optimum in the sense that it minimizes the collective electricity bill (including the cost of both flexible and inflexible demand) and  fills the overnight demand valley. 
\dar{As a result, a rich body of literature has focused on devising distributed and decentralized schemes that are numerically efficient, and can be used by the grid operator to coordinate the strategies of the agents to a Nash equilibrium \cite{ma2013decentralized,grammatico:parise:colombino:lygeros:14, dario2015aggregative, gan2013optimal, chen2014autonomous,paccagnan2016distributed}.} 
Less attention has been devoted to verify whether the optimality statement made in \cite{ma2013decentralized} is still valid in the presence of more general cost functions, agents heterogeneity and realistic charging constraints (e.g., upper bounds on the instantaneous charging, different charging windows, ramping constraints). 
Nonetheless, this is a fundamental prerequisite for the applicability of the aforementioned coordination schemes.\footnote{
While there are multiple factors impacting the choice of a control scheme, if the Nash equilibria do not posses desirable properties, the grid operator has limited incentive in coordinating the agents to such a strategy profile.
}
\\
\indent
\dar{Following \cite{ma2013decentralized}, the efficiency of the Nash equilibrium has been  studied in \cite{Gonz2015}, under the assumption of linear price functions. Both \cite{ma2013decentralized} and \cite{Gonz2015}  focus on \textit{simplex constraints and homogeneous populations}. 
The homogeneity assumption is relaxed in  \cite{deori2016nash}, where the authors provide similar efficiency results of those in \cite{ma2013decentralized}, but limited to \emph{linear} price functions and in a probabilistic sense.
\emph{Non linear price functions} are considered in \cite{deconvergence}, but the efficiency results pertain to the notion of Wardrop equilibrium and charging constraints are limited to upper bounds.
We observe that all the previous works assume that the price at time $t$ depends only on the consumption at the same time instant.
Finally, we note that \cite{Beaude12} also provides efficiency bounds for charging games, but the setup considered therein is different, in that each agent's decision variable is limited to its starting charging time.}

\dar{The aim of this paper is to provide  efficiency results for the Nash equilibrium of aggregative charging games under different assumptions involving {\it finite populations} of vehicles, {\it general convex constraints}, {\it non linear price functions} and {\it  price dependence on different time instants}}. To do so, we model the charging problem as an aggregative game \cite{jensen2010aggregative}, and study the equilibrium efficiency using the notion of \textit{price of anarchy} ($\poa$). The $\poa$ is a measure introduced in game theory to quantify how much selfish behavior degrades the performance of a given system \cite{koutsoupias1999worst}.  
By definition, $\poa\ge1$ and the closer to $1$ the better the overall performance of the system. The result in \cite{ma2013decentralized} can be equivalently stated as the fact that for homogeneous populations with simplex constraints, the $\poa$ converges to~$1$ as the population size grows.
Our main contributions are:
\begin{enumerate}[leftmargin=*]
\vspace*{-1mm}
\item We show that the $\poa$ for charging games with linear price function (that might however depend on all charging instants) and generic convex constraints always converges to $1$, complementing \cite{ma2013decentralized, Gonz2015, deori2016nash, deconvergence};
\item For charging games with generic convex constraints and nonnegative price function that depends only on the instantaneous demand, we show that the $\poa$ converges to $1$ if the price function is a positive pure monomial (i.e., $\alpha z^k$ for some $\alpha, k>0$). 
On the contrary, if the price function does not have this form, it is possible to construct a sequence of games whose $\poa$ does not converge to $1$. 
\fra{In such cases, we show how results for routing games\cite{roughgarden2003price, correa2004selfish} can be used to bound the asymptotic value of the price of anarchy.} 

\item In all the previous cases we provide an explicit bound connecting the efficiency of the equilibria with the (finite) number of vehicles in the game. To the best of our knowledge, this is the first result providing a bound on $\poa$ as a function of the population size, for charging games with general convex constraints and price functions.

\end{enumerate}


\subsubsection*{\bf \emph{Organization}} Section \ref{sec:PF} includes the game formulation and some preliminary notions. In Section \ref{sec:poalarge}  we define the efficiency metric used throughout this manuscript, and present the main results for linear and nonlinear price functions. Section \ref{sec:ATP} focuses on the application of charging a fleet of electric vehicles. All the proofs are reported in the Appendix. 
\subsubsection*{\bf \emph{Notation}}
$\R^n_{\ge0}$ and $\R^n_{>0}$ denote the elements of $\R^n$ whose components are non negative and positive;
$\mathbb{0}_n\in\mb{R}^{n}$ (resp. $\mathds{1}_n$)  is the column vector of zero (resp. unit) entries.
Given $A\in\mathbb{R}^{n\times n}$ not necessarily symmetric, $A\succ0$ $\Leftrightarrow$ $x^\top A x>0,$ $\forall x\neq 0$. 
 Given $g(x):\mathbb{R}^n \rightarrow \mathbb{R}^m$ we define the matrix $\nabla_{x} g(x) \in \mathbb{R}^{n\times m}$ component-wise as
$[\nabla_x g(x)]_{i,j}\coloneqq \frac{\partial g_j(x)}{\partial x\i}$. 
An operator $F:\mc{K}\subset\R^n \rightarrow \R^n$ is called $\alpha$ strongly monotone if $(F(x)-F(y))^\top (x-y)\ge\alpha||x-y||^2$ for some $\alpha>0$, $\forall x,y\in\mc{K}$; $\mathcal{U}[a,b]$ is the uniform distribution on the real interval $[a, b]$. 
\vspace*{-3mm}
\section{Problem formulation}
\label{sec:PF}
Let us consider a population of $\N$ agents, each choosing an action $x\i\!\in\!\mc{X}\i\!\subseteq\!\mb{R}^n$. Agent $i$ incurs the cost $J\i(x\i,\sigma(x)):\mb{R}^n\times \mb{R}^n\rightarrow \mb{R}$ that depends on his own action $x\i \!\in\! \mc{X}\i$ and on the average action $\sigma(x)\!\coloneqq\!\frac{1}{\N}\!\sum_{j=1}^{\N}x\j$ of the population, as typical of aggregative games~\cite{jensen2010aggregative}. We assume that
\vspace*{-1mm}
\begin{equation} 
J\i(x\i,\sigma(x)) \coloneqq p(\sigma(x)+d)^\top x\i\,,
\label{eq:costs}
\vspace*{-1mm}
\end{equation} 
with $d\in\mathbb{R}_{\ge0}^n$ and $p:\R^n\to\R^n$.
The cost in~\eqref{eq:costs} can be used to describe applications where $x^i$ denotes the usage level of a certain commodity,
whose per-unit cost $p$ depends on the average usage level of the  other players plus some inflexible normalized usage level $d$~\cite{ma2013decentralized,chen2014autonomous}. We denote with $\mc{X} \coloneqq \mc{X}^1\times\ldots\times\mc{X}^\N$, and identify such game with the tuple 
\vspace*{-1mm}
\begin{equation}
\label{eq:gameG}
\mathcal{G}\coloneqq\{\N,\{\mathcal{X}^i\}_{i=1}^\N,p\}. 
\vspace*{-4mm}
\end{equation}
\subsection{Nash, Wardrop equilibrium and social optimizer}
\vspace*{-1mm}
We consider two notions of equilibrium for the game $\mc{G}$. 
\vspace*{-2mm}
\begin{definition}[Nash equilibrium \cite{nash1950equilibrium}]\label{def:NE}
A set of actions $x\NE = [x^1\NE; \dots; x^\N\NE] \in \R^{\N n}$ is a Nash equilibrium of the game $\mathcal{G}$  if $x\NE\in\mc{X}$ and for all $ i\in\{1,\dots,\N\}$ and all $ x\i \in\mc{X}\i$ 
\vspace*{-1mm}
\begin{equation}
 J\i(x\i\NE,\sigma(x\NE)) \le J\i\biggl( x\i,\frac 1\N x\i  +  \frac 1\N \sum_{j \neq i} x\j\NE \biggr ). 
\label{eq:def_NE}
\vspace*{-2mm}
\end{equation}
\end{definition}
Observe that on the right-hand side of \eqref{eq:def_NE} the variable $x^i$ appears in both arguments of $J^i(\cdot,\cdot)$.
As the population size grows, the contribution of an agent to the average decreases. This motivates the definition of Wardrop equilibrium.
\vspace*{-1mm}
\begin{definition}[Wardrop equilibrium \cite{wardrop1952road,Gentilearxiv17}]\label{def:WE}
A set of actions $x\WE = [x^1\WE; \dots; x^\N\WE] \in \R^{\N n}$ is a Wardrop equilibrium of $\mathcal{G}$  if $x\WE\in\mc{X}$, and for all $i\in\{1,\dots,\N\}$, and all $x\i\in\mc{X}\i$,
\vspace*{-1mm}
\begin{equation}
J\i(x\i\WE,\sigma(x\WE)) \le J\i ( x\i,\sigma(x\WE) )\,.
 \label{eq:def_WE}
 \vspace*{-2mm}
\end{equation}
\end{definition}
\noindent Note that in this latter definition the average is fixed to $\sigma(x\WE)$ on both sides of \eqref{eq:def_WE}.
\noindent 
Consequently, a feasible set of actions is a Wardrop equilibrium if no agent can improve his cost, assuming that the average action is \emph{fixed}.
\begin{definition}[Social optimizer]
A set of actions $x\SO = [x^1\SO; \dots; x^\N\SO] \in \R^{\N n}$ is a social optimizer of $\mathcal{G}$  if $x\SO\in\mc{X}$ and it minimizes the cost 
$
J\SO(\sigma(x))\coloneqq p(\sigma(x)+d)^\top(\sigma(x)+d).
$
\end{definition}
\noindent Note that the cost $J\SO$ 
 is the sum of all the players costs, divided by $\N$, and the additional term $p(\sigma(x)+d)^\top d$. The reason why the latter term is included is that we want to compute the  total cost of buying the commodity both for the flexible ($\sigma(x)$) and inflexible ($d$) users. 
 This cost was first introduced in \cite{ma2013decentralized} and then used in \cite{Gonz2015, deori2016nash, deconvergence}.
 The inflexible usage level is sometimes modeled in the literature \cite{deori2016nash} as an additional player with constraint set represented by $\{x\in\R^n\mid x=d\cdot \N\}$, where $d$ is the \emph{normalized} inflexible demand. We do not follow such approach here because we are interested in large populations and  this set is unbounded as $\N\rightarrow\infty$. 
 Throughout the manuscript, we denote with $
\s\coloneqq \bigl\{z\in \R^n ~|~ z\!=\!\frac{1}{\N}\sum_{j=1}^\N x^j,~x^j\!\in\!\mc{X}^j,~\forall~j=1,\dots,\N \bigr\}.
$
 \begin{assumption}
 \label{A1}
For $i\in\{1,\dots,\N\}$, the constraint set $\mathcal{X}\i$ is closed, convex, non empty.
For $z\in \s$, the function $z\mapsto p(z+d)$ is continuously differentiable and strongly monotone while $z\mapsto p(z+d)^\top(z+d)$ is strongly convex.
\vspace*{-1mm}
 \end{assumption}
 
  \noindent We denote with $L\SO$, $L_p$ the Lipschitz constant of $J\SO(\cdot)$, $p(\cdot)$, and with $\alpha$ the monotonicity constant of $p(\cdot)$. 
 \vspace*{-4mm}
 \section{Price of Anarchy for finite and large populations}
 \label{sec:poalarge}
 In this section we study the efficiency of equilibria as a function of the population size $\N$.
To do so, we consider a sequence of games $(\mc{G}_\N)_{\N=1}^\infty$ of increasing population size. For fixed $\N$, the game $\mc{G}_\N$ is played amongst $\N$ agents and is defined as in \eqref{eq:gameG} with arbitrary sets $\{\mc{X}\i\}_{i=1}^{\N}$. The function $p$ is instead the same for every game of the sequence.
 
\begin{assumption}
\vspace*{-1mm}
\label{ass:sequence}
There exists a convex, compact set $\mathcal{X}_0\subset\R^n$ s.t. $\cup_{i=1}^\N \mathcal{X}^i\subseteq{\mathcal{X}_0}$ for each game $\mc{G}_\N$ in $(\mc{G}_\N)_{\N=1}^\infty$. Moreover, $J^i(x^i,\sigma(x))$ is convex in $x^i\in\mathcal{X}^i$ for all fixed $x^{-i}\in\mathcal{X}^{-i}$, for all $i\in\{1,\dots,\N\}$. We let $R\coloneqq \max_{y\in\mc{X}_0}||y||$. 
\vspace*{-1mm}
\end{assumption}
 For a given a game $\mc{G}_\N$, we quantify the efficiency of equilibrium allocations using the notion of price of anarchy~\cite{koutsoupias1999worst} 
 \begin{equation}\label{eq:poa}
 \poa_\N \coloneqq \frac{\max_{x_N\in \textup{NE}_\N}J\SO(\sigma(x_N)) }{J\SO(\sigma(x\SO))}\,,
\end{equation}
 where $\textup{NE}_\N\subseteq \mc{X}$ is the set of Nash equilibria of $\mc{G}_\N$ and $x\SO$ is a social optimizer of $\mc{G}_\N$. The price of anarchy captures the ratio between the cost at the worst Nash equilibrium and the optimal cost; by definition $\poa_\N\ge1$. 
 \fra{In the next  subsections we study the behavior of $\poa_\N$, for three different classes of admissible price functions $p$ (of increasing generality).}
 \vspace*{-4mm}
 \subsection{Linear price function}
Throughout this subsection we consider linear  price functions $p$, as detailed in the following.
\vspace*{-1mm}
\begin{assumption}
\label{ass:lin}
The price function $p$ is linear, that is $p(z+d)=C(z+d)$, with $C=C^\top\in\R^{n\times n}$, $C\succ0$.
\vspace*{-1mm}
\end{assumption}
Note that Assumption \ref{ass:lin} implies  strong monotonicity of $z\mapsto p(z+d)$ and  strong convexity of $z\mapsto p(z+d)^\top(z+d)$, therefore Assumption \ref{ass:lin} is consistent with Assumption \ref{A1}. It is easy to verify that $J^i(x^i,\sigma(x))$ is convex in $x^i$, consistently with Assumption \ref{ass:sequence}. 
Nevertheless, $C$ is not required to be diagonal as it was instead in \cite{Gonz2015,deori2016nash}.
\vspace*{-1mm}
\begin{theorem}[$\poa_\N$ bound and convergence to 1]
\label{thm:lin}
~
\vspace*{-5mm}
\begin{itemize}
	\item[a)]{Under Assumption \ref{A1} and \ref{ass:lin}, for any game $\mc{G}_\N$ in the sequence, every Wardrop equilibrium $x\WE$ is a social optimizer i.e. $J\SO(\sigma(x\WE))\le J\SO(\sigma(x)),~\forall x\in \mc{X}$.
	} 
	\item[b)]{With the further Assumption \ref{ass:sequence}, for any fixed game $\mc{G}_\N$ in the sequence it holds that 
	\begin{equation}
	\textstyle
	J\SO(\sigma(x\SO))\le J\SO(\sigma(x\NE))\le J\SO(\sigma(x\SO))+c/{\sqrt{M}}\,,
	\label{eq:boundjlin}
	\end{equation}
	with $c=RL\SO\sqrt{2L_p\alpha^{-1}}$ constant, $x\SO$ social optimizer.\\
	{Thus, if there exists $\hat J\ge 0$ s.t. $J\SO(\sigma(x\SO))>\hat J$ for every game in the sequence $(\mc{G}_\N)_{\N=1}^\infty$, one has}
	{\[
	1\le \poa_\N\le 1+c/\bigl(\hat J\sqrt{\N}\bigl)~~~\text{and}~~
	\lim_{\N\to\infty}\poa_\N=1\,.\]}}
\end{itemize}
\end{theorem}
The proof is reported in the Appendix.
\vspace*{-1mm}
\begin{remark}
The previous theorem extends the results of \cite{ma2013decentralized,Gonz2015,deori2016nash,deconvergence}  simultaneously allowing for arbitrary convex constraints, finite populations, and non diagonal price function. Note that the condition $J\SO(\sigma(x\SO))>\hat J\ge0$ is merely technical and required to properly define $\poa_\N$. This condition is trivially satisfied in the applications when, e.g., every agent requests an amount of charge bounded away from zero. Even if the latter condition does not hold, the cost at any Nash equilibrium converges to the minimum cost as $\N\!\to\!\infty$, see~\eqref{eq:boundjlin}.
\end{remark}

%
\vspace*{-5mm}
\subsection{Non linear homogeneous price function}
In this section we consider $p(z+d)$ to be a nonlinear function, and assume its $t$-th component to depend only on the $t$-th component $z_t+d_t$, for all $t\in\{1,\dots,n\}$. This models, e.g., cases where the unit cost of electricity at every instant of time depends on the total consumption at that same instant.
\begin{assumption}
\label{ass:nonlin}
The price function $p$ takes the form 
\vspace*{-1mm}
\[
p(z+d)=
\begin{bmatrix}
f(z_1+d_1),\hdots,
f(z_n+d_n)
\end{bmatrix}^\top,
\vspace*{-1mm}
\]
with $f(y):\R_{>0}\rightarrow\R_{>0}$.
Further $\mc{X}\i\subseteq \R^n_{\ge0}$ and $d\in\R^n_{>0}$\,.
\end{assumption}
If $f(y)$ is not linear, a simple check shows that, in general, $\nabla_{x^j}(\nabla_{x^i} J^i(x^i,\sigma(x)))\neq\nabla_{x^i}(\nabla_{x^j} J^j(x^j,\sigma(x)))$ when $i\neq j$. Consequently,  the game is not potential, \cite[Theorem 1.3.1]{facchinei2007finite}. Hence methods to bound the $\poa$ based on the existence of an underlying potential function \cite{Gonz2015, deori2016nash}, can not be used here.
\begin{theorem}[$\poa_\N$ convergence and counterexample]
\label{thm:polypoa}
Suppose that Assumptions \ref{A1}, \ref{ass:sequence} and \ref{ass:nonlin} hold. Further assume that $J\SO(\sigma(x\SO))$ $>$$\hat J$ for some $\hat J \ge 0$, for every game in $(\mc{G}_\N)_{\N=1}^\infty$. 
\begin{itemize}
\item[a)] If $f(y)=\alpha y^k$ with $\alpha>0$ and $k>0$, it holds 
\[
	1\le \poa_\N\le 1+c/\bigl(\hat J\sqrt{\N}\bigl)~~~\text{and}~~
	\lim_{\N\to\infty}\poa_\N=1\,,
	\vspace*{-2mm}
	\]
	\vspace*{2mm}
	with $c = RL\SO\sqrt{2 L_p \alpha^{-1}}$ constant.
\item[b)] For $n\ge 2$, if $f(y)$ satisfies the assumptions, but does not take the form $\alpha y^k$ for some $\alpha>0$ and $k>0$, it is possible to construct a sequence of games $(\mc{G}_\N)_{\N=1}^\infty$ for which $\lim_{\N\to\infty}\poa_\N>1$.
\end{itemize}
\end{theorem}
The proof is reported in the Appendix. Therein, the counterexample relative to b) is constructed using $\mc{X}\i=\bar{\mc{X}}$. In other words our  impossibility result holds also for the case of 
homogeneous  populations. This is not in contrast with the result in \cite{ma2013decentralized} or \cite{deconvergence}, because therein the sets $\bar{\mc{X}}$ were assumed to be simplexes with upper bounds constraints. Here we claim that there exists a convex set $\bar{\mc{X}}$ (not a simplex with upper bounds) such that $\poa_\N$ does not converge to $1$. 
\begin{remark}
	The previous theorem is of fundamental importance from the standpoint of the system operator, in that it suggests the use of monomial price functions to guarantee the highest achievable efficiency (all Nash equilibria become social optimizers for large $\N$). If different price functions are chosen, it is always possible to construct a problem instance such that the worst Nash equilibrium is \emph{not} a social optimizer.
\end{remark}
\vspace*{-4mm}
\subsection{Nonlinear heterogeneous price function}
\fra{
In the previous subsection we showed that if the price function is not a monomial, then  $\poa_\N$ may not converge to one. In this section we derive upper bounds for $\poa_\N$ when the price function belongs to a general class of functions and may  be different at different time instants, as formalized next.
 \begin{assumption}
\label{ass:nonlin2}
The price function $p$ takes the form 
\[
p(z+d)=
\begin{bmatrix}
l_1(z_1+d_1),
\hdots,
l_n(z_n+d_n)
\end{bmatrix}^\top,
\]
where $l_t(y):\R_{\ge 0}\rightarrow\R_{\ge0}$, $ l_t\in \mathcal{L}$ for all $t$ and $\mathcal{L}$ is a given class of continuous and nondecreasing price functions.  Further let $\mc{X}\i\subseteq \R^n_{\ge0}$ be non empty, closed and convex.
\end{assumption}
Note that Assumption  \ref{ass:nonlin2} is \emph{less restrictive} than Assumption  \ref{ass:nonlin} as we let the price $l_t$ depend on the time instant $t$. The key idea in this case is to show that  standard results derived in  \cite{roughgarden2003price}, \cite{correa2004selfish} for Wardrop equilibria in routing games can be applied to charging games too. The resulting bounds on $\poa_\N$ can then be derived using the converging result in \cite{Gentilearxiv17}.
 More formally, given a charging game $\mc{G}_{\N}$, we consider an equivalent nonatomic routing game over a parallel network with as many links as charging intervals. To present our next result we introduce the following quantity from \cite[Eq 3.8]{correa2004selfish}
$$\beta(\mathcal{L}):=\sup_{l\in\mathcal{L}}  \sup_{v\ge 0} \left( \frac{1}{vl(v)}\max_{w\ge 0} [ (l(v)-l(w))w] \right).$$
 It follows from \cite{correa2004selfish} that  $\beta(\mathcal{L})\le 1$ and   $[ 1- \beta(\mathcal{L}) ]^{-1}=\alpha(\mathcal{L})$, where  $\alpha(\mathcal{L})$ is the anarchy value for class $\mathcal{L}$ as defined in \cite{roughgarden2003price}.
Therein (see Table 1),  $\alpha(\mathcal{L})$ is computed for classes of functions such as affine, quadratic, polynomials. The key idea of the following theorem is to show that the games considered here are $(1,\beta(\mathcal{L}))$-smooth, as defined in \cite{roughgarden2009intrinsic}.
\begin{theorem}[$\poa_\N$ for heterogeneous price function]\label{thm:routing}
a) Suppose that Assumption \ref{ass:nonlin2} holds. Then for any fixed game $\mathcal{G}_M$ and any Wardrop equilibrium $x_W$ it holds
\begin{equation}\label{eq:stepThm3}
J_S(\sigma(x_W))\le J_S(\sigma(x_S))\alpha(\mathcal{L})
\end{equation}
b) Further suppose Assumptions \ref{A1}, \ref{ass:sequence} hold, and there exists $\hat J\ge 0$ s.t. $J\SO(\sigma(x\SO))>\hat J$ for every game in $(\mc{G}_\N)_{\N=1}^\infty$. 
  Then, for any game $\mathcal{G}_M$ in the sequence
\[
J_S(\sigma(x_S))\le J_S(\sigma(x_N))\le J_S(\sigma(x_S))\alpha(\mathcal{L})+c/\sqrt{M},
\] 
and $1\le\poa_M\le \alpha(\mathcal{L})+c/\bigl(\hat J\sqrt{\N}\bigl),$ thus implying 
$\lim_{M\rightarrow \infty} \poa_M\le \alpha(\mathcal{L}),$
with $c = RL_s\sqrt{2L_p\alpha^{-1}}$.
\end{theorem}
\vspace*{-3mm}
\begin{remark}
If $\mathcal{L}$ contains constant functions, then  \eqref{eq:stepThm3} is tight (see \cite{roughgarden2003price} and  the simulation section). 
This is not a contradiction of Theorems \ref{thm:lin}, \ref{thm:polypoa} because therein either constant functions are not allowed or the price function is assumed to be time independent.   Theorems \ref{thm:lin}, \ref{thm:polypoa} can be seen as refinements of Theorem~\ref{thm:routing} and guarantee that $\lim_{M\rightarrow \infty} \poa_\N=1$ by restricting the admissible class of price functions.
\end{remark}
}
\vspace*{-4mm}
\section{Application to charging of electric vehicles}
\label{sec:ATP}
We consider a population of $\N$  electric vehicles, where the level of charge of vehicle $i$ at time $t$ is described by $s\i_t$. 
Its evolution is specified by the discrete-time system
$
s\i_{t+1} = s\i_t + b\i x\i_t \,,  t = 1, \dots, n$,
where $x\i_t$ is the charging control and $b\i > 0$ is the charging efficiency.
We assume that $x\i_t$ is non-negative, that it cannot exceed $\tilde x^i_t \ge 0$ at time $t$ and that the absolute value of the difference between $x\i_{t}$ and $x\i_{t+1}$ is bounded by $r_i$.
The final level of charge is constrained to $s_{n+1}^i\ge\eta\i$, where $\eta\i \ge 0$ is the desired level of charge of agent $i$.
Denoting $x\i =[x\i_1, \dots, x^i_n]^\top \in \R^n$, the constraints of agent $i$ reduce to%
\vspace*{-1mm}
\begin{equation}
\label{eq:vehicle_constraint}
x\i \!\!\in\! \mc{X}\i \!=\!\! \left\{\!x\i\! \!\in\! \mathbb{R}^n \! \left|\!\!
\!\begin{array}{l}
0 \le x\i_t \le \tilde x\i_t, ~~~~~ \forall \, t=1,\dots,n \\ 
\sum_{t=1}^{n} x\i_t \ge \theta\i\\
|x\i_{t+1}-x\i_{t}|\le r\i,   \forall \, t=1,\dots,n\!-\!1
\end{array}
\!\!\!\!\right.
\right\}
\vspace*{-1mm}
\end{equation}
where $\theta\i \coloneqq {(b\i)}^{-1} (\eta\i - s\i_1)$, with
$s\i_1 \ge 0$ the level of charge for $t=1$. Note that the vehicles are \textit{heterogeneous} in the \textit{total amount of energy} required $\theta^i$ as well as  the \textit{time-varying upper bounds} $\tilde x^i_t$ (that  can be used to model deadlines, availability for charging), and the {\it ramping constraints $r^i$}. Such constraints satisfy Assumption \ref{A1}. Further, we assume that there exists $\hat \eta>0$ such that for each $M$ and $i\in\{1,\dots,\N\}$, $\eta^i\le \hat \eta$ so that $\mathcal{X}_0$  is compact as required in Assumption \ref{ass:sequence}. Note that this is without loss of generality in any practical scenario. The cost function of each vehicle reads as 
\vspace*{-1mm}
\begin{equation}
\textstyle J\i(x\i,\sigma(x))\!=\!\sum_{t=1}^n p_t \left( \frac{\sigma_t(x)+d_t }{\kappa_t}  \right) x\i_t \!= \!p(\sigma(x)+d)^\top x^i,
\label{eq:PEV_energy_bill}
\vspace*{-1mm}
\end{equation}
where we assumed that the energy price for each time interval $p_t:\R_{\ge0}\rightarrow \R_{>0}$ depends on the ratio between total consumption and total capacity $(\sigma_t(x)+d_t)/ \kappa_t$, where $d_t$ and $\sigma_t(x):=\frac{1}{\N}\sum_{i=1}^\N x^i_t$ are the non-EV and EV demand at time $t$ divided by $\N$
and $\kappa_t$ is the total production capacity divided by $\N$ as in~\cite[eq. (6)]{ma2013decentralized}.
To sum up, we define the  game $\mc{G}^\text{EV}_\N$ as in~\eqref{eq:gameG}, with $\mathcal{X}\i$ and $J\i(x\i,\sigma(x))$ as in \eqref{eq:vehicle_constraint} and \eqref{eq:PEV_energy_bill} respectively.
Let $x:=[x^1;\ldots; x^\N]$ be the vector of charging schedules for the whole population. The social cost of the game is $
 J_S(\sigma(x))$$=$$\textstyle \sum_{t=1}^n p_t \left( \frac{\sigma_t(x)+d_t}{\kappa_t}  \right) (\sigma_t(x)+d_t) $$=$$ p(\sigma(x)+d)^\top (\sigma(x)+d)$,
that is, the overall electricity bill for the sum of non-EV and EV demand; $n=24$. For the numerical study, we consider four cases as described next.

\emph{Case $1$.} We set $p_t(y)=0.15y^3$ and choose $\tilde x^i_t$ to allow charging in $[t^i_{\textup{min}},t^i_{\textup{max}}]$, with $t^i_{\textup{min}},t^i_{\textup{max}}$ uniformly randomly
distributed between 5pm and 10am; 
$\theta^i\sim\mathcal{U}[5, 15]$, $r^i\sim\mathcal{U}[1,7]$ and $d_t$ as in \cite[Figure 1]{ma2013decentralized}.

\dar{
\emph{Cases $2$-$4$.} We set $p_t(y)=0.15$ from 5pm to 1am and $p_t(y)=0.15y$ from 2am to 10am. For all vehicles, we choose $\tilde x^i_t$ to allow charging from 5pm to 10am. There are no ramping constraints. Cases 2-4 differ in $\theta^i$, $d_t$ as in the following table.}
\vspace*{-5mm}
\begin{table}[h!]
\centering
\dar{
\begin{tabular}{|c|c|c|}
\hline
 Case& $\theta^i$ & $d_t$ \\ \hline
2 & $9$ & $\mathbb{0}_n$\\
3 &$9$ & as in \cite[Figure 1]{ma2013decentralized} \\
4 &$\mathcal{U}[5,13]$ & $\mathbb{0}_n$ \\ \hline
\end{tabular}
}
\end{table}%

\noindent For each case, we report the (numerical) price of anarchy as a function of $\N$ in Figure \ref{fig:poa_and_diff} (top).
Observe that case $1$ and $4$ feature heterogenous charging needs. For these cases, we have randomly extracted $100$ games $\mc{G}^{\text{EV}}_\N$ (for any fixed $\N$) and report the worst $\poa$ amongst the $100$ realization.
In order to plot the price of anarchy, we computed the ratio between \emph{one} (instead of the \emph{worst}) Nash equilibrium of $\mc{G}^\text{EV}_\N$ and the social optimum. This choice is imposed by the fact that computing all Nash equilibria of $\mc{G}^\text{EV}_\N$ is in general a hard problem.\footnote{To compute a Nash equilibrium we applied the extragradient  algorithm \cite{facchinei2007finite}, which is not guaranteed to converge for small $\N$ as the operator associated with the variational inequality of the Nash problem is not guaranteed to be strongly monotone \cite{Gentilearxiv17}. We thus verified a posteriori that the point where the algorithm stopped was a Nash equilibrium.}
\dar{In Figure \ref{fig:poa_and_diff} (bottom) we plot the difference between the cost at the Nash and at the social optimizer, relative to case~1.}
\vspace*{-2mm}
\newlength\figureheight 
\newlength\figurewidth 
\setlength\figureheight{2.8cm} 
\setlength\figurewidth{0.8\linewidth} 
\begin{figure}[h!]
     \begin{subfigure}[b]{\linewidth}
          \centering
          \resizebox{1\linewidth}{!}{\input{poa.tikz}} 
     \end{subfigure}
    \\[0.1cm]
     \begin{subfigure}[b]{\linewidth}
          \centering
          \resizebox{1\linewidth}{!}{\input{whisker_plot.tikz}}
     \end{subfigure}
     \vspace*{-2mm}
     \caption{Price of anarchy (top), and \dar{cost difference between Nash and social optimum (bottom)} as a function of $\N$.}
     \label{fig:poa_and_diff}
     \vspace*{-5mm}
\end{figure}
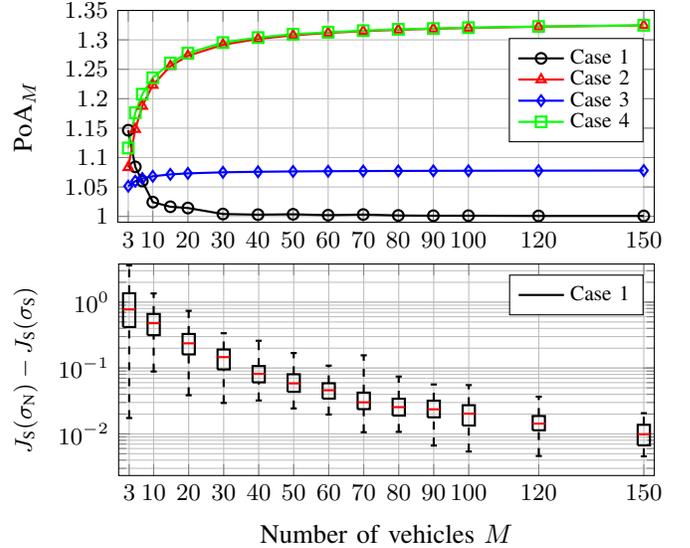

Thanks to the choice of parameters and price function, case 1 meets the  Assumptions \ref{A1}, \ref{ass:sequence} and \ref{ass:nonlin}  (see Lemma \ref{lem:ass} in the Appendix). Thus, Theorem \ref{thm:polypoa}b) guarantees that $\lim_{M\rightarrow \infty} \poa_M=1$. The numerical results reported in Figure~\ref{fig:poa_and_diff} (top, black line) are consistent with it: the ratio between the cost at the Nash and the cost at the social optimum converges to one. \dar{In addition to this, Figure~\ref{fig:poa_and_diff} (bottom) shows that also the difference between these costs converges to zero, as guaranteed by the proof of theorem \ref{thm:polypoa}a) and the boundedness of $\mc{X}_0$.
A typical plot describing the valley filling property of the equilibrium in case 1 can be found e.g., in \cite[Figure 2]{ma2013decentralized}.
Case 2 has been constructed so that the corresponding Wardrop equilibrium features the worst possible asymptotic price of anarchy within the class of affine cost functions (for which $\alpha(\mathcal{L})=4/3$, see \cite{roughgarden2003price}).
The numerics of Figure \ref{fig:poa_and_diff} (top, red line) show that $\poa_\N$ converges to $1.33\approx 4/3=\alpha(\mathcal{L})$. Cases 3 and 4 are a modification of case 2. 
While the presence of base demand (case 3) helps in lowering the price of anarchy, the impact of heterogeneity (case 4) on the asymptotic price of anarchy is minor (blue and green plots in Figure~\ref{fig:poa_and_diff}).}

\vspace*{-3mm}
\section{Conclusions} 
\vspace*{-1mm}
We considered the problem of charging a fleet of heterogeneous electric vehicles as formulated using game theoretic tools. More precisely, we studied the efficiency of the resulting equilibrium allocations, measured by the concept of price of anarchy.
We showed that the price of anarchy converges to one as the population of vehicles grow
if the price function is linear (but possibly dependent on all the time instants), or if the price function depends only on the instantaneous demand and is a positive pure monomial. \dar{We provided efficiency bounds for general non linear functions.} For these three cases, we also provided bounds on the $\poa$ as a function of the population size. Our theoretical findings are corroborated by means of numerical simulations.
\dar{We conclude noting that the question regarding the efficiency of equilibria in aggregative games is of interest for a broader class of cost functions than those studied here (e.g., quasi convex costs). We leave this as a future work.}
%
\vspace*{-5mm}
\section*{Appendix A: Characterization of the average} 
 \vspace*{-1mm}
\noindent 
This section characterizes the average players' action $\sigma(x)$ at the Wardrop equilibrium and at the social optimizer of $\mathcal{G}$. 
%
\begin{definition}[Variational inequality~\cite{facchinei2007finite}]
\label{def:vi}
Given $\mathcal{K}\subseteq \mathbb{R}^\ell$ and $F:\mathcal{K}\rightarrow \mathbb{R}^\ell$. A point $\bar x\in\mathcal{K}$ is a solution of the variational inequality $\textup{VI}(\mathcal{K},F)$ if $\,\forall x\in\mathcal{K}$, $F(\bar x)^\top (x-\bar x)\ge 0.$
\end{definition} 
 \begin{lemma}[Equivalent characterizations]
 \label{lemma:averageVI}
 \begin{enumerate}[leftmargin=*]
 \item[]
 \item[] \hspace*{-5.8mm} Suppose  Assumption \ref{A1} holds.
 \item
 Given $x\WE$ a Wardrop equilibrium, its average $\sigma(x\WE)$ solves $\textup{VI}(\s,F\WE)$, with $F\WE:\mathbb{R}^n\rightarrow\mathbb{R}^n$, 
$F\WE(z) \coloneqq p(z+d)$.
The  $\textup{VI}(\s,F\WE)$ admits a unique solution $\sigma\WE$. Let us define $\mc{X}\WE\coloneqq\{x\in\mc{X}~\text{s.t.}~\frac{1}{\N}\sum_{j=1}^\N x\j=\sigma\WE\}$. Then any vector of strategies $x\WE\in\mc{X}\WE$ is a Wardrop equilibrium.
\item
Given $x\SO$ a social optimizer, its average $\sigma(x\SO)$ solves $\textup{VI}(\s,F\SO)$, with $F\SO:\mathbb{R}^n\rightarrow\mathbb{R}^n$, 
$F\SO(z)\coloneqq p(z+d)+[\nabla_z p(z+d)](z+d).$
The $\textup{VI}(\s,F\SO)$ admits a unique solution $\sigma\SO$. Define $\mc{X}\SO\coloneqq\{x\in\mc{X}~\text{s.t.}~\frac{1}{\N}\sum_{j=1}^\N x\j=\sigma\SO\}$. Then any vector of strategies $x\SO\in\mc{X}\SO$ is a social optimizer.
\end{enumerate}
 \end{lemma}
 \begin{proof}
 {\bf 1)}
  The sets $\mc{X}\i$ are convex and closed by Assumption \ref{A1}; further, for fixed $z\in \s$, the functions $J\i(x\i,z)$ are linear and thus convex in $x\i\in\mc{X}\i$ for all $i\in\{1,\dots,\N\}$. It follows that (see \cite{Gentilearxiv17}) a Wardrop equilibrium $x\WE$ satisfies 
  \vspace*{-1mm}
  \begin{equation}
  [\mathds{1}_{\N}\otimes \,p(\sigma(x\WE)+d)]^\top \!(x-x\WE)\!\ge\!0,~~\forall x\!\in\!\mc{X}.
  	\label{eq:bigvi}
  	\vspace*{-1mm}
  \end{equation}
 %
  Rearranging and dividing by $\N$  we get
$
p(\sigma(x\WE)+d)^\top(\frac{1}{\N}\sum_{j=1}^\N x\i-\frac{1}{\N}\sum_{j=1}^\N x\WE\i)\ge 0,	
$
for all $x\in\mathcal{X}$, or equivalently
$p(\sigma(x\WE)+d)^\top(z-\sigma(x\WE)  )\ge 0,~\forall z\in\s,$
  that is, $\sigma(x\WE)$ solves $\textup{VI}(\s,F\WE)$.
   By Assumption \ref{A1} $F\WE(z)=p(z+d)$ is strongly monotone and $\s$ is closed, convex (since the sets $\mathcal{X}^i$ are closed, convex), hence by \cite{facchinei2007finite} $\textup{VI}(\s,F\WE)$ has a unique solution $\sigma\WE$. By definition of variational inequality, for any $z\in \s$ it holds $p(\sigma\WE+d)^\top(z-\sigma\WE)\ge0$. By definition of $x\WE\in\mc{X}\WE$, we have $\sigma(x\WE)=\sigma\WE$. It follows that $p(\sigma(x\WE)+d)^\top(z-\sigma(x\WE))\ge0$ for any $z\in \s$.  By definition of $\s$, we conclude that \eqref{eq:bigvi} holds for all $x\in\mc{X}$. Thus, $x\WE$ is a Wardrop equilibrium (see \cite{Gentilearxiv17}).
\newline
{\bf 2)} By Assumption \ref{A1}, the set $\mc{X}$ is convex and closed and $J\SO(\sigma(x))$ is  convex. Hence, a social optimizer $x\SO$ satisfies
 \begin{equation}
  \label{eq:bigvi2}
\nabla_x[p(\sigma(x)+d)(\sigma(x)+d)]_{|x={x\SO}}^\top (x-x\SO)\ge0
  ~~\forall x\in\mc{X}\,.
  \end{equation}
 Note that $M \nabla_{x^i} (p(\sigma(x)+d)^\top(\sigma(x)+d)) =p(\sigma(x\SO)+d)+[\nabla_z p(\sigma(x\SO)+d)](\sigma(x\SO)+d)$ for all $i\in\{1,\ldots,M\}$. Consequently, \eqref{eq:bigvi2} is equivalent to
$
  [ p(\sigma(x\SO)+d)+[\nabla_z p(\sigma(x\SO)+d)](\sigma(x\SO)+d)]^\top 
  (\sigma(x)-\sigma(x\SO))\ge0\,,
$
 that is $\sigma(x\SO)$ solves $\textup{VI}(\s,F\SO)$. The remaining claims are proven similarly to 1).
\end{proof}
\vspace*{-3mm}
\section*{Appendix B: Proofs of Theorem \ref{thm:lin}, \ref{thm:polypoa} and \ref{thm:routing}}
%
%
\begin{myproof1}
\newline
{\bf a)} Let $x\WE$ be a Wardrop equilibrium. By Lemma \ref{lemma:averageVI} part 1, $\sigma(x\WE)$ solves $\textup{VI}(\s,F\WE)$. Because of Assumption \ref{ass:lin}, $F\SO(z)=C(z+d)+C^\top(z+d)=2C(z+d)=2F\WE(z)$. Since the two operators $F\WE(z)$ and $F\SO(z)$ are parallel for each $z\in\s$, it follows from the definition of variational inequality that  $\sigma(x\WE)$ must solve $\textup{VI}(\s,F\SO)$ too. Using Lemma \ref{lemma:averageVI} part 2 we conclude that $x\WE$ must be a social optimizer. \\
{\bf b)} By definition $J\SO(\sigma(x\SO))\le J\SO(\sigma(x\NE))$ and so $1\le \poa_\N$.
Further, Assumption \ref{ass:sequence} and the strong monotonicity of $p(z+d)$ (Assumption \ref{A1}) allow us to use the convergence result of \cite[Theorem 1]{Gentilearxiv17}. That is, for any Nash equilibrium $x\NE$ and Wardrop equilibrium $x\WE$ of the game $\mc{G}_\N$, 
$
||\sigma(x\WE)-\sigma(x\NE)||\le \sqrt{2 R^2L_p \alpha^{-1}{\N}^{-1}}.$ It follows that $|J\SO(\sigma(x\NE))- J\SO(\sigma(x\WE))|\le L\SO\sqrt{2 R^2L_p \alpha^{-1}{\N}^{-1}} = c\sqrt{M^{-1}}.$ Since every Wardrop equilibrium is socially optimum (previous point of this proof), one has $|J\SO(\sigma(x\NE))- J\SO(\sigma(x\SO))|\le c\sqrt{M^{-1}}$ and thus
$J\SO(\sigma(x\NE))\le J\SO(\sigma(x\SO))+c\sqrt{M^{-1}}$. The final result regarding the price of anarchy follows from the latter inequality upon dividing both sides by $J\SO(\sigma(x\SO))>\hat J\ge 0$.
\end{myproof1}

\vspace*{2mm}

\begin{myproof2}
\newline
{\bf a)}
We first show that any Wardrop equilibrium is a social optimizer.
To do so, observe that the function $f(y)=\alpha y^k$ satisfies all the assumptions required by Lemma \ref{lemma:averageVI} (see Lemma \ref{lem:ass} in the Appendix).
Let $x\WE$ be a Wardrop equilibrium of $\mc{G}_\N$. By Lemma \ref{lemma:averageVI}, $\sigma(x\WE)$ solves $\textup{VI}(\s,F\WE)$. Thanks to Assumption \ref{ass:nonlin} and the choice of $f(y)$,
\[F\SO(z)\!=\!(k+1)
[
\alpha(z_1+d_1)^k,
\hdots,
\alpha(z_n+d_n)^k
]^\top
\!\!=\!(k+1)F\WE(z)\,.
\]

Hence $\sigma(x\WE)$ solves $\textup{VI}(\s,F\SO)$ too. Using Lemma \ref{lemma:averageVI} we conclude that $x\WE$ must be a social optimizer. 
The proof is now identical to the proof of Theorem \ref{thm:lin}, part b).
\newline
{\bf b)}
  If $f(y)$ does not take the form $\alpha y^k$ for some $\alpha>0$ and $k>0$, by Lemma \ref{lemma:notaligned} there exists a point $\bar z\in\R^n_{>0}$ for which $F\WE(\bar z)$ and $F\SO(\bar z)$ are not aligned, i.e. for which $F\SO(\bar z)\neq h F\WE(\bar z)$ for all $h\in\R$. 
We intend to construct a sequence of games $\mc{G}_\N$ so that for every $\mc{G}_\N$ in the sequence the unique average at the Wardrop equilibrium is exactly $\bar z$, that is  $\bar z$ solves $\textup{VI}(\s,F\WE)$, but $\bar z$ does not solve $\textup{VI}(\s,F\SO)$. This fact indeed proves, by Lemma \ref{lemma:averageVI}, that for any game $\mc{G}_\N$  the Wardrop equilibria of $\mc{G}_\N$ are not social minimizers. 
 Since $\sigma(x\NE)\to\sigma(x\WE)$ as $\N\to\infty$ \cite[Theorem 1]{Gentilearxiv17}, one concludes that $\poa$ cannot converge to~$1$.

In the following we construct a sequence of games with the above mentioned properties.
To this end let us define $\mc{X}\i\coloneqq\bar{\mc{X}}\subseteq \R^n$, so that $\s=\bar{\mc{X}}$ with
 $ \bar{\mc{X}}\coloneqq\{\bar z+\alpha v_1  +\beta v_2~~ \alpha,\beta\in[0~1]\}\cap\R^n_{\ge0},$
  where $v_1\coloneqq\bar F\WE$, $v_2\coloneqq(\bar F\WE^\top\bar F\SO )\bar F\WE- (\bar F\WE^\top \bar F\WE)\bar F\SO $ and $\bar F\WE\coloneqq F\WE(\bar z)$, $\bar F\SO\coloneqq F\SO(\bar z)$; see Figure \ref{fig:setX}. The intuition is that $-v_2$ is the component of $\bar F_S$ that lives in the same plane as $\bar F_S$ and $\bar F_W$ and is orthogonal to $\bar F_W$, so that $\bar F_W^\top v_2=0$.
  \begin{figure}[h!]
        \centering
        \includegraphics[scale=0.8]{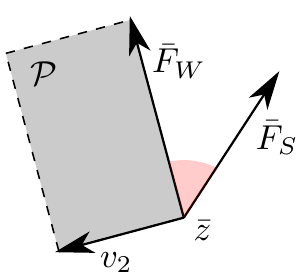}
        \vspace*{-2mm}
        \caption{Construction of the set $\bar{\mc{X}}$.}
        \label{fig:setX}
        \vspace*{-7mm}
   \end{figure}
    Observe that $\s=\bar{\mc{X}}$ is the intersection of a bounded and convex set with the positive orthant and thus satisfies Assumptions \ref{A1}, \ref{ass:sequence} and \ref{ass:nonlin}.
  It is easy to verify that $\bar z \in \bar{\mc{X}}$ and that $F\WE(\bar z)^\top(z-\bar z)=\alpha ||F\WE(\bar z)||^2\ge0$ for all $z\in \s=\bar{\mc{X}}$, so that $\bar z$ solves $\textup{VI}(\s,F\WE)$. Let us pick $\hat z=\bar z+ \beta v_2$. Note that since $\bar z>0$, for $\beta$ small enough  $\hat z$ belongs to $\R^n_{>0}$ as well and thus to $\bar{\mc{X}}$. Then $F\SO(\bar z)^\top(\hat z-\bar z)=\beta (\bar F\SO^\top\bar F\WE)^2-\beta ||\bar F\SO||^2||\bar F\WE||^2< 0$. The inequality is strict because $\bar F\WE$, $\bar F\SO$ are neither parallel nor zero (Lemma \ref{lemma:notaligned}). 
Thus, $\bar z$ does not solve $\textup{VI}(\s,F\SO)$.
\end{myproof2}
\begin{lemma}
\label{lemma:notaligned}
For $n\ge 2$, if $f(y)$ satisfies Assumptions \ref{A1}, \ref{ass:sequence} and \ref{ass:nonlin}, but does not take the form $\alpha y^k$ for some $\alpha>0$ and $k>0$, then there exists $\bar z \in\R^n_{>0}$ such that $F\SO(\bar z)\neq h F\WE(\bar z)$, $\forall h\in\R$. Moreover, $F\SO(\bar z)\neq 0$, $F\WE(\bar z)\neq 0$. 
\end{lemma}
\vspace*{-1.5mm}
\begin{proof}
Let us consider the first statement.
By contradiction, assume there exists $\beta(z):\R^n_{>0}\rightarrow\R$ such that $F\SO(z)= \beta(z) F\WE(z)$ for all $z \in\R^n_{>0}$. This implies
\vspace*{-0.5mm}
\begin{equation}
\label{eq:alphas}
f'(z_t+d_t)(z_t+d_t)=(\beta(z_1,\dots,z_n)-1)f(z_t+d_t)\,,
\vspace*{-0.5mm}
\end{equation}

for all $t\in\{1,\dots,n\}$ and for all $z \in\R^n_{>0}$, $d\in\R^n_{>0}$. 
By Assumption \ref{ass:nonlin}, $f(z_t+d_t)>0$.
Hence one can divide \eqref{eq:alphas} for $f(z_t+d_t)$ without loss of generality, and conclude that $\beta(z_1,\dots,z_n)=\beta_1(z_1)=\dots=\beta_n(z_n)$ with $\beta_i:\R\rightarrow\R$ for all $z\in\R^n_{>0}$.
For $n\ge2$ the last condition implies $\beta(z_1,\dots,z_n)=b$ constant. Equation \eqref{eq:alphas} reads as $
f'(y)y=(b-1)f(y)~~\forall y>0$, 
whose continuously differentiable solutions are all and only $f(y)=a y^{b-1}$. Note that if $a\le0$ or $b\le 1$, Assumption \ref{A1} is not satisfied, while if $a>0$ and $b>1$ we contradicted the assumption that $f(y)$ did not take the form $\alpha y^k $ for some $\alpha>0$ and $k>0$. 
Setting $h=0$ in the previous claim gives $F\SO(\bar z)\ne0$.
Since $f:\R_{>0}\rightarrow\R_{>0}$, one has $F\WE(\bar z):=[ f(\bar z_t+d_t)]_{t=1}^n\neq0$.
\end{proof}
\vspace*{-1mm}
\begin{lemma}\label{lem:ass}
Suppose that the price function $p$ is as in Assumption \ref{ass:nonlin} with $f(y)=\alpha y^k$, $\alpha>0,k>0$. Then $p$ satisfies Assumption \ref{A1} and \ref{ass:sequence}.
\end{lemma}
\begin{proof}
Note that $\nabla_z p(z+d)$ is a diagonal matrix with entry $f'(z_t+d_t)$ in position $(t,t)$.  Since $f'(y)=\alpha k y^{k-1}>0$ for all $y>0$ and since $z_t+d_t$ is positive by assumption for all $t$, we get that $p(z+d)$ is continuously differentiable and that $\nabla_z p(z+d)\succ 0$ i.e. that $z\mapsto p(z+d)$ is strongly monotone. Similarly, one can show that the Hessian of $p(z+d)^\top (z+d)$ and the Hessian of $J^i(x^i,\sigma(x))$ with respect to $x\i$ are positive definite. Thus, $z\mapsto p(z+d)^\top (z+d)$ and $x^i \mapsto J^i(x^i,\sigma(x))$ are strongly convex. See \cite{paccagnan2018efficiency} for further details.
%
%
\end{proof}
\vspace*{2mm}
\begin{myproof3} 
\fra{
\newline
We prove only  {a)} as  {b)} can be shown as in Theorem~1b).
We define
$C^{\sigma_1}(\sigma_2):=p(\sigma_1+d)^\top(\sigma_2+d)$
so that $J_S(\sigma)=C^{\sigma}(\sigma)$.
Let $x_W$ be any Wardrop equilibrium. Then, the average $\bar \sigma:=\sigma_W$ solves VI$(\Sigma,F_W)$ i.e.
$F_W(\bar \sigma)^\top(\sigma-\bar \sigma)\ge 0,~ \forall \sigma\in\Sigma.$
This can be seen following the proof of Lemma \ref{lemma:averageVI} part 1), and observing that only convexity and closedness of $\mc{X}^i$ are required.
Equivalently,
$  J_S(\bar \sigma) \le C^{\bar \sigma}(\sigma),~\forall \sigma\in\Sigma.$
However,
$C^{\bar \sigma}(\sigma)=\sum_{t} l_t(\bar \sigma_t+d_t) (\sigma_t+d_t)
=J_S(\sigma)+ \sum_{t} [l_t(\bar \sigma_t+d_t) - l_t( \sigma_t+d_t)] (\sigma_t+d_t) 
=J_S(\sigma)+ \sum_{t} \frac{[l_t(v_t) - l_t(w_t)]w_t  }{l_t(v_t)v_t}   l_t(v_t) v_t 
\le J_S(\sigma)+ \sum_{t} \beta(\mathcal{L})   l_t(v_t) v_t 
=  J_S(\sigma)+  \beta(\mathcal{L}) J_S(\bar\sigma)
$
where we use $v_t:=\bar \sigma_t+d_t\ge d_t$, $w_t:= \sigma_t+d_t\ge d_t$ and $d_t\ge 0$.
The previous relation holds for all $\sigma\in\Sigma$. Selecting $\sigma=\sigma_S$ (the optimum average), we get
$ J_S(\bar \sigma)\! \le \!J_S(\sigma_S)\!+\!  \beta(\mathcal{L}) J_S(\bar\sigma).$
Rearranging we obtain \eqref{eq:stepThm3}.
}
\end{myproof3}

\vspace*{-2mm}
\bibliographystyle{IEEEtran}
\bibliography{References}

\end{document}

%% file: poa.tikz
%
%
\begin{tikzpicture}

\begin{axis}[%
width=\figurewidth,
height=\figureheight,
at={(1.011111in,0.641667in)},
scale only axis,
xmin=0,
xmax=153,
xmajorgrids,
xtick={3, 10, 20, 30, 40 , 50, 60, 70, 80, 90, 100, 120, 150},
ymin=0.99,
ymax=1.35,
ytick={1, 1.05, 1.10, 1.15, 1.20, 1.25, 1.30, 1.35},
ymajorgrids,
tick label style={font=\small},
ylabel={$\poa_\N$},
legend style={at={(0.98,0.87)},anchor=north east, row sep=-3pt},
]

\addplot [color=black,solid,mark=o,mark options={solid},line width = 0.8pt]
  table[row sep=crcr]{%
3	1.146440320202752\\
5	1.084144206845967\\
7	1.060048849309918\\
10	1.024232150863788\\
15	1.016372686491602\\
20	1.014257536279855\\
30	1.004045564995405\\
40	1.002948190905122\\
50	1.003581851865993\\
60  1.002083234902703\\
70	1.003119089620615\\
80  1.001584796583014\\
90	1.001117731557795\\
100	1.001180347871256\\
120 1.000782030475284\\
150 1.000777440384723\\ 
};
\addlegendentry{\footnotesize Case 1}

\addplot [color=red,solid,mark=triangle,mark options={solid},line width = 0.8pt]
  table[row sep=crcr]{%
3	1.083333271041991\\
5	1.148147949310946\\
7	1.187499715409593\\
10	1.223140306615917\\
15	1.255208121274737\\
20	1.272864593882444\\
30	1.291709937733426\\
40	1.301606129192305\\
50	1.307702065977234\\
60 	1.311833603676443\\
70	1.314818263630529\\
80  1.317075457595903\\
90  1.318841734023247\\
100	1.320261976515319\\
120 1.322404157939918\\
150 1.324560203528079\\ 
};
\addlegendentry{\footnotesize Case 2}

\addplot [color=blue,solid,mark=diamond,mark options={solid},line width = 0.8pt]
  table[row sep=crcr]{
3	1.051328159113605\\
5	1.059896867679024\\
7	1.064396255666115\\
10	1.068183912171345\\
15	1.071412416249527\\
20	1.073131774140623\\
30	1.074924501049531\\
40	1.075850149130869\\
50	1.076415196445974\\
60  1.076795962142782\\
70	1.077069926099951\\
80  1.077268560416244\\
90  1.077437745707750\\
100	1.077567147972555\\
120 1.077749962830918\\
150 1.077942306748157\\ 
};
\addlegendentry{\footnotesize Case 3}

\addplot [color=green,solid,mark=square,mark options={solid},line width = 0.8pt]
  table[row sep=crcr]{
3	1.116195356881952\\
5	1.176091225026618\\
7   1.207709846007565 \\
10	1.235579439042558\\
15	1.260437393215178\\
20	1.277438692108011\\
30	1.295324853591825\\
40	1.303891968957424\\
50	1.309541572633434\\
60  1.312817040389996\\
70	1.315639950431473\\
80  1.317885457093774\\
90  1.319246690488960\\
100	1.320741186315272\\
120 1.322749824875607\\
150 1.324731164796894\\
};
\addlegendentry{\footnotesize Case 4}

\end{axis}
\end{tikzpicture}%

%% file: whisker_plot.tikz
%
%

\begin{tikzpicture}

\begin{axis}[%
width=\figurewidth,
height=\figureheight,
at={(1.011111in,0.641667in)},
scale only axis,
xlabel={Number of vehicles $\N$},
ylabel={\small $J\SO(\sigma\NE)-J\SO(\sigma\SO)$},
xmin=0,
xmax=153,
xtick={3,10,20,30,40,50,60,70,80,90,100,120,150},
ymode=log,
ymin=0,
ymax=3.78925706869876,
tick label style={font=\small},
xmajorgrids,
ymajorgrids,
yminorgrids,
legend style={at={(0.98,0.95)},anchor=north east},
]
\addplot [color=black,dashed,forget plot,line width = 0.8pt]
  table[row sep=crcr]{%
3	1.36455483005518\\
3	3.60913300466812\\
};
\addplot [color=black,dashed,forget plot,line width = 0.8pt]
  table[row sep=crcr]{%
10	0.661267334431301\\
10	1.35553933809437\\
};
\addplot [color=black,dashed,forget plot,line width = 0.8pt]
  table[row sep=crcr]{%
20	0.328574005167333\\
20	0.738532377630726\\
};
\addplot [color=black,dashed,forget plot,line width = 0.8pt]
  table[row sep=crcr]{%
30	0.191085695817939\\
30	0.337240888073737\\
};
\addplot [color=black,dashed,forget plot,line width = 0.8pt]
  table[row sep=crcr]{%
40	0.107120495362501\\
40	0.259150726619048\\
};
\addplot [color=black,dashed,forget plot,line width = 0.8pt]
  table[row sep=crcr]{%
50	0.0804700267636278\\
50	0.168761845885427\\
};
\addplot [color=black,dashed,forget plot,line width = 0.8pt]
  table[row sep=crcr]{%
60	0.0586946310393692\\
60	0.108704704080303\\
};
\addplot [color=black,dashed,forget plot,line width = 0.8pt]
  table[row sep=crcr]{%
70	0.0423517383944656\\
70	0.155860370240291\\
};
\addplot [color=black,dashed,forget plot,line width = 0.8pt]
  table[row sep=crcr]{%
80	0.0342026095864441\\
80	0.0742559089782659\\
};
\addplot [color=black,dashed,forget plot,line width = 0.8pt]
  table[row sep=crcr]{%
90	0.0319720661540686\\
90	0.0562498390783759\\
};
\addplot [color=black,dashed,forget plot,line width = 0.8pt]
  table[row sep=crcr]{%
100	0.0272606068318311\\
100	0.0550346794052956\\
};
\addplot [color=black,dashed,forget plot,line width = 0.8pt]
  table[row sep=crcr]{%
120	0.0187099504006731\\
120	0.0367291686001678\\
};
\addplot [color=black,dashed,forget plot,line width = 0.8pt]
  table[row sep=crcr]{%
150	0.00457224675071188\\
150	0.00669404844279597\\
};
\addplot [color=black,dashed,forget plot,line width = 0.8pt]
  table[row sep=crcr]{%
3	0.0174042884840873\\
3	0.419050712738169\\
};
\addplot [color=black,dashed,forget plot,line width = 0.8pt]
  table[row sep=crcr]{%
10	0.0884143618951185\\
10	0.315071352216037\\
};
\addplot [color=black,dashed,forget plot,line width = 0.8pt]
  table[row sep=crcr]{%
20	0.0385989068193169\\
20	0.160788076930675\\
};
\addplot [color=black,dashed,forget plot,line width = 0.8pt]
  table[row sep=crcr]{%
30	0.0294740455648679\\
30	0.0951942826321641\\
};
\addplot [color=black,dashed,forget plot,line width = 0.8pt]
  table[row sep=crcr]{%
40	0.032208380073186\\
40	0.0606422563625166\\
};
\addplot [color=black,dashed,forget plot,line width = 0.8pt]
  table[row sep=crcr]{%
50	0.0243749374216549\\
50	0.043505726175141\\
};
\addplot [color=black,dashed,forget plot,line width = 0.8pt]
  table[row sep=crcr]{%
60	0.0196724536211192\\
60	0.0343265440014449\\
};
\addplot [color=black,dashed,forget plot,line width = 0.8pt]
  table[row sep=crcr]{%
70	0.0106149258284347\\
70	0.023774650159055\\
};
\addplot [color=black,dashed,forget plot,line width = 0.8pt]
  table[row sep=crcr]{%
80	0.0107322584015108\\
80	0.0189465062578833\\
};
\addplot [color=black,dashed,forget plot,line width = 0.8pt]
  table[row sep=crcr]{%
90	0.00665172405531678\\
90	0.0178710946771581\\
};
\addplot [color=black,dashed,forget plot,line width = 0.8pt]
  table[row sep=crcr]{%
100	0.00540093012132559\\
100	0.0134064602108701\\
};
\addplot [color=black,dashed,forget plot,line width = 0.8pt]
  table[row sep=crcr]{%
120	0.00462583521935045\\
120	0.0113589686794029\\
};
\addplot [color=black,dashed,forget plot,line width = 0.8pt]
  table[row sep=crcr]{%
150	0.0137577346121773\\
150	0.0206127677327004\\
};
\addplot [color=black,solid,forget plot,line width = 0.8pt]
  table[row sep=crcr]{%
2.125	3.60913300466812\\
3.875	3.60913300466812\\
};
\addplot [color=black,solid,forget plot,line width = 0.8pt]
  table[row sep=crcr]{%
9.125	1.35553933809437\\
10.875	1.35553933809437\\
};
\addplot [color=black,solid,forget plot,line width = 0.8pt]
  table[row sep=crcr]{%
19.125	0.738532377630726\\
20.875	0.738532377630726\\
};
\addplot [color=black,solid,forget plot,line width = 0.8pt]
  table[row sep=crcr]{%
29.125	0.337240888073737\\
30.875	0.337240888073737\\
};
\addplot [color=black,solid,forget plot,line width = 0.8pt]
  table[row sep=crcr]{%
39.125	0.259150726619048\\
40.875	0.259150726619048\\
};
\addplot [color=black,solid,forget plot,line width = 0.8pt]
  table[row sep=crcr]{%
49.125	0.168761845885427\\
50.875	0.168761845885427\\
};
\addplot [color=black,solid,forget plot,line width = 0.8pt]
  table[row sep=crcr]{%
59.125	0.108704704080303\\
60.875	0.108704704080303\\
};
\addplot [color=black,solid,forget plot,line width = 0.8pt]
  table[row sep=crcr]{%
69.125	0.155860370240291\\
70.875	0.155860370240291\\
};
\addplot [color=black,solid,forget plot,line width = 0.8pt]
  table[row sep=crcr]{%
79.125	0.0742559089782659\\
80.875	0.0742559089782659\\
};
\addplot [color=black,solid,forget plot,line width = 0.8pt]
  table[row sep=crcr]{%
89.125	0.0562498390783759\\
90.875	0.0562498390783759\\
};
\addplot [color=black,solid,forget plot,line width = 0.8pt]
  table[row sep=crcr]{%
99.125	0.0550346794052956\\
100.875	0.0550346794052956\\
};
\addplot [color=black,solid,forget plot,line width = 0.8pt]
  table[row sep=crcr]{%
119.125	0.0367291686001678\\
120.875	0.0367291686001678\\
};
\addplot [color=black,solid,forget plot,line width = 0.8pt]
  table[row sep=crcr]{%
149.125	0.0206127677327004\\
150.875	0.0206127677327004\\
};
\addplot [color=black,solid,forget plot,line width = 0.8pt]
  table[row sep=crcr]{%
2.125	0.0174042884840873\\
3.875	0.0174042884840873\\
};
\addplot [color=black,solid,forget plot,line width = 0.8pt]
  table[row sep=crcr]{%
9.125	0.0884143618951185\\
10.875	0.0884143618951185\\
};
\addplot [color=black,solid,forget plot,line width = 0.8pt]
  table[row sep=crcr]{%
19.125	0.0385989068193169\\
20.875	0.0385989068193169\\
};
\addplot [color=black,solid,forget plot,line width = 0.8pt]
  table[row sep=crcr]{%
29.125	0.0294740455648679\\
30.875	0.0294740455648679\\
};
\addplot [color=black,solid,forget plot,line width = 0.8pt]
  table[row sep=crcr]{%
39.125	0.032208380073186\\
40.875	0.032208380073186\\
};
\addplot [color=black,solid,forget plot,line width = 0.8pt]
  table[row sep=crcr]{%
49.125	0.0243749374216549\\
50.875	0.0243749374216549\\
};
\addplot [color=black,solid,forget plot,line width = 0.8pt]
  table[row sep=crcr]{%
59.125	0.0196724536211192\\
60.875	0.0196724536211192\\
};
\addplot [color=black,solid,forget plot,line width = 0.8pt]
  table[row sep=crcr]{%
69.125	0.0106149258284347\\
70.875	0.0106149258284347\\
};
\addplot [color=black,solid,forget plot,line width = 0.8pt]
  table[row sep=crcr]{%
79.125	0.0107322584015108\\
80.875	0.0107322584015108\\
};
\addplot [color=black,solid,forget plot,line width = 0.8pt]
  table[row sep=crcr]{%
89.125	0.00665172405531678\\
90.875	0.00665172405531678\\
};
\addplot [color=black,solid,forget plot, line width = 0.8pt]
  table[row sep=crcr]{%
99.125	0.00540093012132559\\
100.875	0.00540093012132559\\
};
\addplot [color=black,solid,forget plot, line width = 0.8pt]
  table[row sep=crcr]{%
119.125	0.00462583521935045\\
120.875	0.00462583521935045\\
};
\addplot [color=black,solid,forget plot, line width = 0.8pt]
  table[row sep=crcr]{%
149.125	0.00457224675071188\\
150.875	0.00457224675071188\\
};
\addplot [color=black,solid,line width = 0.8pt]
  table[row sep=crcr]{%
1.25	0.419050712738169\\
1.25	1.36455483005518\\
4.75	1.36455483005518\\
4.75	0.419050712738169\\
1.25	0.419050712738169\\
};
\addlegendentry{\footnotesize Case 1}
\addplot [color=black,solid,forget plot,line width = 0.8pt]
  table[row sep=crcr]{%
8.25	0.315071352216037\\
8.25	0.661267334431301\\
11.75	0.661267334431301\\
11.75	0.315071352216037\\
8.25	0.315071352216037\\
};
\addplot [color=black,solid,forget plot,line width = 0.8pt]
  table[row sep=crcr]{%
18.25	0.160788076930675\\
18.25	0.328574005167333\\
21.75	0.328574005167333\\
21.75	0.160788076930675\\
18.25	0.160788076930675\\
};
\addplot [color=black,solid,forget plot,line width = 0.8pt]
  table[row sep=crcr]{%
28.25	0.0951942826321641\\
28.25	0.191085695817939\\
31.75	0.191085695817939\\
31.75	0.0951942826321641\\
28.25	0.0951942826321641\\
};
\addplot [color=black,solid,forget plot,line width = 0.8pt]
  table[row sep=crcr]{%
38.25	0.0606422563625166\\
38.25	0.107120495362501\\
41.75	0.107120495362501\\
41.75	0.0606422563625166\\
38.25	0.0606422563625166\\
};
\addplot [color=black,solid,forget plot,line width = 0.8pt]
  table[row sep=crcr]{%
48.25	0.043505726175141\\
48.25	0.0804700267636278\\
51.75	0.0804700267636278\\
51.75	0.043505726175141\\
48.25	0.043505726175141\\
};
\addplot [color=black,solid,forget plot,line width = 0.8pt]
  table[row sep=crcr]{%
58.25	0.0343265440014449\\
58.25	0.0586946310393692\\
61.75	0.0586946310393692\\
61.75	0.0343265440014449\\
58.25	0.0343265440014449\\
};
\addplot [color=black,solid,forget plot,line width = 0.8pt]
  table[row sep=crcr]{%
68.25	0.023774650159055\\
68.25	0.0423517383944656\\
71.75	0.0423517383944656\\
71.75	0.023774650159055\\
68.25	0.023774650159055\\
};
\addplot [color=black,solid,forget plot,line width = 0.8pt]
  table[row sep=crcr]{%
78.25	0.0189465062578833\\
78.25	0.0342026095864441\\
81.75	0.0342026095864441\\
81.75	0.0189465062578833\\
78.25	0.0189465062578833\\
};
\addplot [color=black,solid,forget plot,line width = 0.8pt]
  table[row sep=crcr]{%
88.25	0.0178710946771581\\
88.25	0.0319720661540686\\
91.75	0.0319720661540686\\
91.75	0.0178710946771581\\
88.25	0.0178710946771581\\
};
\addplot [color=black,solid,forget plot,line width = 0.8pt]
  table[row sep=crcr]{%
98.25	0.0134064602108701\\
98.25	0.0272606068318311\\
101.75	0.0272606068318311\\
101.75	0.0134064602108701\\
98.25	0.0134064602108701\\
};
\addplot [color=black,solid,forget plot,line width = 0.8pt]
  table[row sep=crcr]{%
118.25	0.0113589686794029\\
118.25	0.0187099504006731\\
121.75	0.0187099504006731\\
121.75	0.0113589686794029\\
118.25	0.0113589686794029\\
};
\addplot [color=black,solid,forget plot,line width = 0.8pt]
  table[row sep=crcr]{%
148.25	0.00669404844279597\\
148.25	0.0137577346121773\\
151.75	0.0137577346121773\\
151.75	0.00669404844279597\\
148.25	0.00669404844279597\\
};
\addplot [color=red,solid,forget plot,line width = 0.8pt]
  table[row sep=crcr]{%
1.25	0.775141399049009\\
4.75	0.775141399049009\\
};
\addplot [color=red,solid,forget plot,line width = 0.8pt]
  table[row sep=crcr]{%
8.25	0.479259094489848\\
11.75	0.479259094489848\\
};
\addplot [color=red,solid,forget plot,line width = 0.8pt]
  table[row sep=crcr]{%
18.25	0.236865156871474\\
21.75	0.236865156871474\\
};
\addplot [color=red,solid,forget plot,line width = 0.8pt]
  table[row sep=crcr]{%
28.25	0.146272097264909\\
31.75	0.146272097264909\\
};
\addplot [color=red,solid,forget plot,line width = 0.8pt]
  table[row sep=crcr]{%
38.25	0.081630208456005\\
41.75	0.081630208456005\\
};
\addplot [color=red,solid,forget plot,line width = 0.8pt]
  table[row sep=crcr]{%
48.25	0.0584347641333203\\
51.75	0.0584347641333203\\
};
\addplot [color=red,solid,forget plot,line width = 0.8pt]
  table[row sep=crcr]{%
58.25	0.0460113833980742\\
61.75	0.0460113833980742\\
};
\addplot [color=red,solid,forget plot,line width = 0.8pt]
  table[row sep=crcr]{%
68.25	0.0301875261668059\\
71.75	0.0301875261668059\\
};
\addplot [color=red,solid,forget plot,line width = 0.8pt]
  table[row sep=crcr]{%
78.25	0.0254812825089523\\
81.75	0.0254812825089523\\
};
\addplot [color=red,solid,forget plot,line width = 0.8pt]
  table[row sep=crcr]{%
88.25	0.0236318299097107\\
91.75	0.0236318299097107\\
};
\addplot [color=red,solid,forget plot,line width = 0.8pt]
  table[row sep=crcr]{%
98.25	0.0202582912901157\\
101.75	0.0202582912901157\\
};
\addplot [color=red,solid,forget plot,line width = 0.8pt]
  table[row sep=crcr]{%
118.25	0.0143398639406094\\
121.75	0.0143398639406094\\
};
\addplot [color=red,solid,forget plot,line width = 0.8pt]
  table[row sep=crcr]{%
148.25	0.00987508357857081\\
151.75	0.00987508357857081\\
};
\end{axis}
\end{tikzpicture}%

%% file: ms.bbl
\begin{thebibliography}{10}
\providecommand{\url}[1]{#1}
\csname url@samestyle\endcsname
\providecommand{\newblock}{\relax}
\providecommand{\bibinfo}[2]{#2}
\providecommand{\BIBentrySTDinterwordspacing}{\spaceskip=0pt\relax}
\providecommand{\BIBentryALTinterwordstretchfactor}{4}
\providecommand{\BIBentryALTinterwordspacing}{\spaceskip=\fontdimen2\font plus
\BIBentryALTinterwordstretchfactor\fontdimen3\font minus
  \fontdimen4\font\relax}
\providecommand{\BIBforeignlanguage}[2]{{%
\expandafter\ifx\csname l@#1\endcsname\relax
\typeout{** WARNING: IEEEtran.bst: No hyphenation pattern has been}%
\typeout{** loaded for the language `#1'. Using the pattern for}%
\typeout{** the default language instead.}%
\else
\language=\csname l@#1\endcsname
\fi
#2}}
\providecommand{\BIBdecl}{\relax}
\BIBdecl

\bibitem{albadi2007demand}
M.~H. Albadi and E.~F. El-Saadany, ``Demand response in electricity markets: An
  overview,'' in \emph{Power Engineering Society General Meeting, 2007}.

\bibitem{ma2013decentralized}
Z.~Ma, D.~S. Callaway, and I.~A. Hiskens, ``Decentralized charging control of
  large populations of plug-in electric vehicles,'' \emph{IEEE Transactions on
  Control Systems Technology}, vol.~21, no.~1, pp. 67--78, 2013.

\bibitem{gan2013optimal}
L.~Gan, U.~Topcu, and S.~H. Low, ``Optimal decentralized protocol for electric
  vehicle charging,'' \emph{IEEE Transactions on Power Systems}, vol.~28,
  no.~2, pp. 940--951, 2013.

\bibitem{grammatico:parise:colombino:lygeros:14}
S.~Grammatico, F.~Parise, M.~Colombino, and J.~Lygeros, ``Decentralized
  convergence to {N}ash equilibria in constrained deterministic mean field
  control,'' \emph{IEEE Transactions on Automatic Control}, vol.~61, no.~11,
  pp. 3315--3329, 2016.

\bibitem{dario2015aggregative}
D.~Paccagnan, M.~Kamgarpour, and J.~Lygeros, ``On aggregative and mean field
  games with applications to electricity markets,'' in \emph{2016 European
  Control Conference (ECC)}, June 2016, pp. 196--201.

\bibitem{chen2014autonomous}
H.~Chen, Y.~Li, R.~H. Louie, and B.~Vucetic, ``Autonomous demand side
  management based on energy consumption scheduling and instantaneous load
  billing: An aggregative game approach,'' \emph{IEEE Transactions on Smart
  Grid}, vol.~5, no.~4, pp. 1744--1754, 2014.

\bibitem{paccagnan2016distributed}
D.~Paccagnan, B.~Gentile, F.~Parise, M.~Kamgarpour, and J.~Lygeros,
  ``Distributed computation of generalized {N}ash equilibria in quadratic
  aggregative games with affine coupling constraints,'' in \emph{Proceedings of
  the IEEE Conference on Decision and Control}, 2016, pp. 6123--6128.

\bibitem{Gonz2015}
M.~Gonzales, S.~Grammatico, and J.~Lygeros, ``On the price of being selfish in
  large populations of plug-in electric vehicles,'' in \emph{Proceedings of the
  IEEE Conference on Decision and Control}, 2015, pp. 6542--6547.

\bibitem{deori2016nash}
L.~Deori, K.~Margellos, and M.~Prandini, ``On the connection between nash
  equilibria and social optima in electric vehicle charging control games,''
  \emph{IFAC-PapersOnLine}, vol.~50, no.~1, pp. 14\,320--14\,325, 2017.

\bibitem{deconvergence}
A.~De~Paola, D.~Angeli, and G.~Strbac, ``Convergence and optimality of a new
  iterative price-based scheme for distributed coordination of flexible loads
  in the electricity market,'' in \emph{Proceedings of the IEEE Conference on
  Decision and Control}, 2017.

\bibitem{Beaude12}
O.~Beaude, S.~Lasaulce, and M.~Hennebel, ``Charging games in networks of
  electrical vehicles,'' in \emph{2012 6th International Conference on Network
  Games, Control and Optimization (NetGCooP)}, Nov 2012, pp. 96--103.

\bibitem{jensen2010aggregative}
M.~K. Jensen, ``Aggregative games and best-reply potentials,'' \emph{Economic
  theory}, vol.~43, no.~1, pp. 45--66, 2010.

\bibitem{koutsoupias1999worst}
E.~Koutsoupias and C.~Papadimitriou, ``Worst-case equilibria,'' in \emph{Annual
  Symposium on Theoretical Aspects of Computer Science}.\hskip 1em plus 0.5em
  minus 0.4em\relax Springer, 1999, pp. 404--413.

\bibitem{roughgarden2003price}
T.~Roughgarden, ``The price of anarchy is independent of the network
  topology,'' \emph{Journal of Computer and System Sciences}, vol.~67, no.~2,
  pp. 341--364, 2003.

\bibitem{correa2004selfish}
J.~R. Correa, A.~S. Schulz, and N.~E. Stier-Moses, ``Selfish routing in
  capacitated networks,'' \emph{Mathematics of Operations Research}, vol.~29,
  no.~4, pp. 961--976, 2004.

\bibitem{nash1950equilibrium}
J.~F. Nash \emph{et~al.}, ``Equilibrium points in n-person games,'' \emph{Proc.
  Nat. Acad. Sci. USA}, vol.~36, no.~1, pp. 48--49, 1950.

\bibitem{wardrop1952road}
J.~G. Wardrop, ``Some theoretical aspects of road traffic research.''
  \emph{Proceedings of the institution of civil engineers}, vol.~1, no.~3, pp.
  325--362, 1952.

\bibitem{Gentilearxiv17}
D.~{Paccagnan}, B.~{Gentile}, F.~{Parise}, M.~{Kamgarpour}, and J.~{Lygeros},
  ``{Nash and Wardrop equilibria in aggregative games with coupling
  constraints},'' \emph{Accepted in Transactions on Automatic Control. arXiv
  preprint arXiv:1702.08789}, Feb. 2017.

\bibitem{facchinei2007finite}
F.~Facchinei and J.~Pang, \emph{{Finite-dimensional variational inequalities
  and complementarity problems}}.\hskip 1em plus 0.5em minus 0.4em\relax
  Springer Science {\&} Business Media, 2007.

\bibitem{roughgarden2009intrinsic}
T.~Roughgarden, ``Intrinsic robustness of the price of anarchy,'' in
  \emph{Proceedings of the forty-first annual ACM symposium on Theory of
  computing}.\hskip 1em plus 0.5em minus 0.4em\relax ACM, 2009, pp. 513--522.

\bibitem{paccagnan2018efficiency}
D.~Paccagnan, F.~Parise, and J.~Lygeros, ``On the efficiency of nash equilibria
  in aggregative charging games,'' \emph{arXiv preprint arXiv:1803.02583},
  2018.

\end{thebibliography}
